\documentclass[a4paper,11pt,leqno]{amsart}
\usepackage{amsmath}
\usepackage{amsfonts}
\usepackage{eucal}
\usepackage{amsthm}
\usepackage{bm}
\numberwithin{equation}{section}

\def\co{\mathbb{C}}
\def\ze{\mathbb{Z}}

\newcommand{\Slash}[1]{{\ooalign{\hfil#1\hfil\crcr\raise.167ex\hbox{/}}}}

\newtheorem{thm}{Theorem}
\newtheorem{lem}[thm]{Lemma}

\theoremstyle{definition}

\theoremstyle{remark}
\newtheorem{rem}{Remark\!}

\begin{document}

\title[The $\ze_2$ Index of Disordered Topological Insulators]{The $\ze_2$ Index of Disordered Topological 
Insulators with Time Reversal Symmetry}
\author[H. Katsura]{Hosho Katsura}
\author[T. Koma]{Tohru Koma}
\address[Katsura]{
Department of Physics, Graduate School of Science, The University of Tokyo, Hongo, Bunkyo-ku, Tokyo 113-0033, JAPAN
}
\email{katsura@phys.s.u-tokyo.ac.jp}
\address[Koma]{Department of Physics, Gakushuin University, 
Mejiro, Toshima-ku, Tokyo 171-8588, JAPAN}
\email{tohru.koma@gakushuin.ac.jp}
\date{\today}

\begin{abstract}
We study disordered topological insulators with time reversal symmetry. 
Relying on the noncommutative index theorem which relates the Chern number 
to the projection onto the Fermi sea and the magnetic flux operator, 
we give a precise definition of the $\ze_2$ index 
which is a noncommutative analogue of the Atiyah-Singer $\ze_2$ index. 
We prove that the noncommutative $\ze_2$ index is robust 
against any time-reversal symmetric perturbation including disorder potentials 
as long as the spectral gap 
at the Fermi level does not close. 
\end{abstract}

\maketitle 

\section{Introduction}
\label{sec:intro}

In a seminal paper, Haldane \cite{Haldane} proposed a tight-binding model of spinless fermions 
on a honeycomb lattice. 
Interestingly, although the net magnetic flux  
through the unit cell of the lattice is vanishing, the model possesses a topological phase 
in which the energy bands carry nontrivial Chern numbers \cite{TKNN,Kohmoto}. 
Kane and Mele \cite{KaneMele} generalized the Haldane model to a spin-1/2 fermion model 
with spin-orbit coupling, so that the model is time-reversal invariant. Because of the time-reversal symmetry, 
the sum of Chern numbers for occupied bands is always vanishing. 
Therefore, at first glance, it appears that the model is topologically trivial. 
Kane and Mele, however, showed that the model exhibits a $\ze_2$ topological order 
which is characterized by the $\ze_2$ topological invariant, instead of the Chern number.  
Further, their study indicates that the $\ze_2$ topological invariant  
they constructed is related to the $\ze_2$ index which was introduced 
by Atiyah and Singer \cite{AtiyahSinger1,AtiyahSinger2}, in order to characterize topological properties 
of a manifold. 
The argument of Kane and Mele relies on the properties of the Bloch wavefunctions, 
and hence requires translational invariance. 
Therefore, one of the important issues is to find a mathematically rigorous and numerically accessible representation 
of the Kane-Mele $\ze_2$ index which 
is robust against perturbations such as disorder potentials.\footnote{
Following the Niu-Thouless-Wu argument \cite{NiuThoulessWu}, Kane and Mele proposed a generalization 
of the $\ze_2$ index to systems with disorder and interactions. 
However, the robustness and the $\ze_2$ quantization of such an index have not yet been proved rigorously.}

We should also remark that Schulz-Baldes \cite{SB,GSB,DNSB} pointed out that the $\ze_2$ index by 
Kane and Mele is equal to the $\ze_2$ index of the Fredholm operator which is written in terms of 
the projection onto the Fermi sea and the magnetic flux operator. 
As is well known, in the case of the standard quantum Hall system, the usual index of the corresponding 
Fredholm operator is equal to the Chern number. This relation is nothing but 
the noncommutative index theorem \cite{BVS,ASS,AG} which is a noncommutative analogue 
of the Atiyah-Singer index theorem. 
In general, such a Fredholm operator is noncompact and its spectrum contains the essential spectrum. 
Thus, it seems hard to prove the continuity of the eigenvalues of the noncompact Fredholm operator 
against generic perturbations such as disorder, 
although the continuity leads to the robustness of the $\ze_2$ index. 
 
In the present paper, instead of the Fredholm operator  \cite{SB,GSB}, 
we use a pair of projections which was introduced by \cite{ASS,ASS2}. 
More precisely, we give an operator-theoretic definition of the $\ze_2$ index 
by relying on the noncommutative index theorem \cite{ASS,ASS2,AG} 
which relates the Chern number to the index of the pair of the projections.  
One of the projections is the projection onto the Fermi sea and 
the other is its unitary deformation which physically implies 
the insertion of a magnetic flux into 
the system \cite{ASS,DNSB}. The advantage of our approach is that the difference between the two projections 
becomes a compact operator. 
Therefore, all of the spectrum are discrete with finite multiplicity except for zero. 
This makes the treatment of perturbations much easier than that for noncompact Fredholm operators. 

In conventional differential geometry \cite{AtiyahSinger1,AtiyahSinger2}, 
the $\ze_2$ index is defined by the analytical index of an operator 
such as an elliptic pseudo-differential operator on a manifold. 
A special property of the multiplicities in the spectrum of the operator yields 
the robustness of the $\ze_2$ index \cite{AtiyahSinger1}. 
On the other hand, in our approach based on noncommutative geometry \cite{Connes}, 
the corresponding $\ze_2$ index is defined by the analytical index of the above compact operator, say $A$, 
consisting of the projection onto the Fermi sea and the magnetic flux operator. 
Similarly to the Kramers degeneracy in the energy spectrum, the time-reversal symmetry leads to 
a special property of the degeneracy in the spectrum of $A$: Every eigenvalue $\lambda$ of $A$ 
shows an even degeneracy when $0<|\lambda|<1$.  
This property yields the robustness of the constructed $\ze_2$ index 
as in the case of the Atiyah-Singer $\ze_2$ index.

The present paper is organized as follows: In the next section, we describe the models we consider. In Sec. 3, we state and prove our main result (Theorem~\ref{thm:main}) about the noncommutative $\ze_2$ index.  
As an example, in Sec.~\ref{KMmodel}, we demonstrate that the $\ze_2$ index for the Kane-Mele model can be computed from the Chern number for the corresponding spinless Haldane model. 
Appendix~\ref{IndChern} is devoted to a short review of noncommutative geometry, in which the relation between the analytical index of an operator and the Chern number is given. 
In Appendix~\ref{ChernHaldane}, we present a calculation of the Chern number of the projection onto 
the Fermi sea for the Haldane model in the present setting. 
In Appendix~\ref{RelationZ2}, we demonstrate that the present $\ze_2$ index is equal to 
that by Kane and Mele for the translationally invariant system. 

\section{Models}
\label{sec:Model}

Consider a tight-binding model of spin-$1/2$ fermions on the square lattice $\ze^2$. 
We denote by $\uparrow$ and $\downarrow$ the spin-up and the spin-down indices, respectively. 
We assume that there are $r$ orbitals at each lattice site.  
For each site ${\bf n}\in\ze^2$, we denote by $\varphi_\alpha({\bf n},\mu)\in\co$ the amplitude of a wavefunction $\varphi\in\ell^2(\ze^2,\co^r\otimes\co^2)$, where $\alpha\in\{\uparrow,\downarrow \}$ and $\mu=1,2,...,r$ denote the spin and the orbital indices, respectively. 
The action of the Hamiltonian $H$ on the wavefunction $\varphi$ is written as 
\begin{equation}
\label{tightbindingHam}
(H\varphi)_\alpha({\bf n},\mu)=\sum_{{\bf m}\in\ze^2}\sum_{\nu=1}^r\sum_{\beta\in\{\uparrow,\downarrow \}}
t_{{\bf n},\mu;{\bf m},\nu}^{\alpha,\beta}\varphi_\beta({\bf m},\nu)
\end{equation}
for ${\bf n}\in\ze^2$, $\mu=1,2,\ldots, r$, and $\alpha\in\{\uparrow,\downarrow \}$, 
where $t_{{\bf n},\mu;{\bf m}\nu}^{\alpha,\beta}$ are the hoping integrals satisfying the hermitian condition,  
$$
\overline{t_{{\bf n},\mu;{\bf m},\nu}^{\alpha,\beta}}=t_{{\bf m},\nu;{\bf n},\mu}^{\beta,\alpha}.
$$  
Here, $\overline{\cdots}$ stands for the complex conjugate. 
In the following, we assume that the range of the hoping integrals is finite. 
The two-component spinor wavefunctions for the $\mu$-th orbital at the site ${\bf n}\in\ze^2$ is written in the form
$$
\varphi({\bf n},\mu)=\left(\begin{matrix} 
\varphi_\uparrow({\bf n},\mu) \\
\varphi_\downarrow({\bf n},\mu)
\end{matrix}
\right).
$$

Let us introduce the time reversal transformation $\Theta$ for wavefunctions as 
\begin{equation}
\label{timereversal}
\varphi^\Theta:=\Theta \varphi=U^\Theta\overline{\varphi}\quad \mbox{for \ } 
\varphi\in\ell^2(\ze^2,\co^r\otimes\co^2),  
\end{equation}
where $U^\Theta$ is a unitary transformation which can be written as
a product of local unitary transformations with a compact support, and 
has a finite period, i.e., invariant under a finite translation 
on the present lattice $\ze^2$.   
As usual, the complex 
conjugates for wavefunctions are defined by 
$$
(\overline{\varphi})_\alpha({\bf n},\mu)=\overline{\varphi_\alpha({\bf n},\mu)}
\quad\mbox{for \ } {\bf n}\in\ze^2, \ \, \mu=1,2,\ldots,r, \ \ \mbox{and}\ \ \alpha\in\{\uparrow,\downarrow \}.
$$
In what follows, we assume that the time-reversal transformation $\Theta$ satisfies
$$
\Theta^2\varphi=-\varphi\quad \mbox{for } \ \varphi\in\ell^2(\ze^2,\co^r\otimes\co^2), 
$$
and say that $\Theta$ is an odd time-reversal transformation.

Let $a$ be an operator on the Hilbert space $\ell^2(\ze^2,\co^r\otimes\co^2)$. 
We say that the operator $a$ is  odd time-reversal
symmetric if the following condition holds: 
\begin{equation}
\Theta(a\varphi)=a\varphi^\Theta\quad \mbox{for any} \ \varphi\in\ell^2(\ze^2,\co^r\otimes\co^2).
\end{equation} 

\medskip

\section{The $\ze_2$ Index}
\label{Sec:Z2Ind}

We assume that the Fermi level lies in a spectral gap of the present Hamiltonian $H$. 
We write $P_{\rm F}$ for the spectral projection on energies below the Fermi energy $E_{\rm F}$. 
We introduce a unitary transformation $U_{\bf a}$ 
for wavefunctions $\varphi\in\ell^2(\ze^2,\co^r\otimes\co^2)$ as 
$$
(U_{\bf a}\varphi)({\bf n},\mu)=U_{\bf a}({\bf n})\varphi({\bf n},\mu)\quad 
\mbox{for}\ \ {\bf n}\in\ze^2, \ \ \mbox{and}\ \ \mu=1,2,\ldots,r,  
$$
where 
\begin{equation}
\label{def:Ua}
U_{\bf a}({\bf u}):=\frac{u_1+iu_2-(a_1+ia_2)}{|u_1+iu_2-(a_1+ia_2)|}
\end{equation}
for ${\bf u}=(u_1,u_2)\in\ze^2$ and ${\bf a}=(a_1,a_2)\in(\ze^2)^\ast$. 
Here, $(\ze^2)^\ast$ denotes the dual lattice of $\ze^2$, i.e., $(\ze^2)^\ast:=\ze^2-(1/2,1/2)$. 

We write  
\begin{equation}
\label{def:A}
A=P_{\rm F}-U_{\bf a} P_{\rm F}U_{\bf a}^\ast
\end{equation}
for short. 
Clearly, this operator $A$ is the difference between two projections. 
As in \cite{ASS,ASS2,AG}, the operator $A^3$ is trace class,\footnote{For trace class operators, 
see, e.g., Chap.~VI.6 of the book \cite{RSI}.} 
and the relative index can be defined as 
\begin{equation}
\label{index}
{\rm Ind}(P_{\rm F},U_{\bf a} P_{\rm F}U_{\bf a}^\ast) 
:={\rm dim}\;{\rm ker}\;(A-1)
-{\rm dim}\;{\rm ker}\;(A+1)={\rm Tr}\; A^3,
\end{equation}
where ${\rm dim}~{\rm ker}~O$ stands for the dimension of the kernel of an operator $O$. 
This implies that the quantity ${\rm Tr}\; A^3$ is quantized to an integer. 
Actually, this is nothing but the Chern number. See Appendix~\ref{IndChern} for details. 

Consider a Hamiltonian $H$ of the form (\ref{tightbindingHam}) which is odd time-reversal symmetric.
We will show that the $\ze_2$ index for such a Hamiltonian $H$ can be defined by 
\begin{equation}
\label{ze2index}
{\rm Ind}_2(P_{\rm F},U_{\bf a} P_{\rm F}U_{\bf a}^\ast) 
:={\rm dim}\;{\rm ker}\;(P_{\rm F}-U_{\bf a} P_{\rm F}U_{\bf a}^\ast-1)\ {\rm modulo}\; 2.
\end{equation}
More precisely, we will show that the $\ze_2$ index constructed above is robust 
against any perturbation which has the same odd time-reversal symmetry 
as the unperturbed Hamiltonian whenever the Fermi level lies in the spectral gap of the total Hamiltonian.   

We write 
\begin{equation}
\label{def:B}
B=1-P_{\rm F}-U_{\bf a}P_{\rm F}U_{\bf a}^\ast
\end{equation}
for short. Then, as in \cite{ASS2}, the following two relations hold:  
\begin{equation}
\label{ABcommutator}
AB+BA=0\quad \mbox{and}\quad A^2+B^2=1.
\end{equation}
Here, the operator $A$ is given by (\ref{def:A}). 

We note that the spectrum of the operator $A$ of (\ref{def:A}) is 
discrete with finite multiplicity except for zero because $A$ is self-adjoint and 
$A^3$ is trace class.  

Following \cite{ASS2}, 
we use the standard supersymmetry argument 
to prove that many of the eigenvectors of the operator $A$ 
come in pairs related by the operator $B$. 
Let $\varphi$ be an eigenvector of $A$ with 
eigenvalue $\lambda\in(0,1)$, i.e.,
$A\varphi=\lambda \varphi$, with $\lambda \in (0,1)$. 
Then, the anticommutation relation of (\ref{ABcommutator}) yields  
$$
AB\varphi=-BA\varphi=-\lambda B\varphi.
$$
Further, the second relation of (\ref{ABcommutator}) yields 
$$
B^2\varphi=(1-A^2)\varphi=(1-\lambda^2)\varphi.
$$
These results imply 
that $B$ is an invertible map from an eigenvector with eigenvalue $\lambda\in(0,1)$ to that with $-\lambda$. 
Thus, there is a one-to-one correspondence between the states $\varphi$ and $B \varphi$, 
and their eigenvalues come in pairs $\pm \lambda$, provided that $0<|\lambda|<1$. 

Next, we show that the time-reversal transformation $\Theta$ 
plays a similar role as the operator $B$ in the above argument. 
Let $\varphi$ be an eigenvector of $A$ with a strictly positive eigenvalue, 
i.e., $A\varphi=\lambda \varphi$, with $\lambda >0$. 
{From} (\ref{def:A}), one has 
\begin{equation}
\label{eq:TRA}
\Theta(P_{\rm F}-U_{\bf a}P_{\rm F}U_{\bf a}^\ast)\varphi=\lambda \varphi^\Theta.
\end{equation}
Since the Hamiltonian $H$ is time-reversal symmetric, 
the spectral projection $P_{\rm F}$ is time-reversal symmetric, too. 
We assume that the unitary operator $U^\Theta$ 
in the time-reversal transformation $\Theta$ of (\ref{timereversal}) satisfies 
\begin{equation}
\label{assump:UTheta}
U^\Theta U_{\bf a}(U^\Theta)^\ast=U_{\bf a}. 
\end{equation}
Under these assumptions, Eq.~(\ref{eq:TRA}) can be rewritten as
$$
(P_{\rm F}-U_{\bf a}^\ast P_{\rm F}U_{\bf a})\varphi^\Theta=\lambda \varphi^\Theta.
$$
Further, the left-hand side of the above equation is written as
$$
U_{\bf a}^\ast(U_{\bf a}P_{\rm F}U_{\bf a}^\ast-P_{\rm F})U_{\bf a}\varphi^\Theta
=-U_{\bf a}^\ast AU_{\bf a}\varphi^\Theta. 
$$
Consequently, we have 
$$
A(U_{\bf a}\varphi^\Theta)=-\lambda U_{\bf a}\varphi^\Theta. 
$$
This implies that $U_{\bf a}\varphi^\Theta$ is an eigenvector of $A$ with eigenvalue $-\lambda$. Furthermore, 
the transformation is invertible. 

\begin{lem}
\label{lem:commute} 
Let $\varphi$ be an eigenvector of $A$ with eigenvalue $\lambda$ satisfying $0<|\lambda|<1$. 
Then, the maps, $B$ and $U_{\bf a}\Theta$, are commuting with each other when acting on the vector $\varphi$, i.e., 
\begin{equation}
\label{commutUThetaB}
U_{\bf a}(B\varphi)^\Theta =B(U_{\bf a}\varphi^\Theta). 
\end{equation}
\end{lem}

\begin{proof}
{From} the definition (\ref{def:B}) of the operator $B$ and the assumption (\ref{assump:UTheta}), one has 
\begin{align*}
\Theta(B\varphi)&=(1-P_{\rm F}-U_{\bf a}^\ast P_{\rm F}U_{\bf a})\varphi^\Theta \\
&=U_{\bf a}^\ast(1-U_{\bf a}P_{\rm F}U_{\bf a}^\ast-P_{\rm F})U_{\bf a}\varphi^\Theta
=U_{\bf a}^\ast B U_{\bf a}\varphi^\Theta.
\end{align*}
Therefore, we have (\ref{commutUThetaB}). 
\end{proof}

Clearly, two vectors $\varphi$ and $U_{\bf a}(B\varphi)^\Theta$ 
are eigenvectors of the operator $A$ with the same eigenvalue $\lambda$.  
If these two vectors are linearly independent of each other, then the corresponding sector which is 
spanned by the two vectors is a two-dimensional space. 

\begin{lem}
Let $\varphi$ be an eigenvector of $A$ with  
eigenvalue $\lambda$ satisfying $0<|\lambda|<1$. 
Then, 
\begin{equation}
\label{orthovarphiUBThetavarphi}
\langle\varphi,U_{\bf a}(B\varphi)^\Theta\rangle=0
\end{equation}
\end{lem}

\begin{proof}
For a vector $\psi\in\ell^2(\ze^2,\co^r\otimes\co^2)$, one has 
\begin{equation}
\label{ThetapsiThetavarphiinn}
\langle \Theta\psi,\Theta\varphi\rangle=\langle\overline{\psi},\overline{\varphi}\rangle
=\langle\varphi,\psi\rangle
\end{equation}
{from} the definition of the time-reversal transformation $\Theta$. We set 
$$
\psi=U_{\bf a}(B\varphi)^\Theta.
$$
Substituting this into the right-hand side of the above equation (\ref{ThetapsiThetavarphiinn}), 
we have 
\begin{equation}
\label{ThetapsiThetavarphiinn2}
\langle \Theta\psi,\Theta\varphi\rangle=\langle \varphi,U_{\bf a}(B\varphi)^\Theta\rangle.
\end{equation}
On the other hand, the above vector $\psi$ can be written as 
$$
\psi=\Theta(U_{\bf a}^\ast B\varphi), 
$$
in the same way as in the proof of Lemma~\ref{lem:commute}. 
Substituting this into the left-hand side of Eq.~(\ref{ThetapsiThetavarphiinn2}), we obtain
\begin{equation}
\label{ThetapsiThetavarphiinn3}
-\langle U_{\bf a}^\ast B\varphi, \Theta\varphi\rangle=\langle \varphi,U_{\bf a}(B\varphi)^\Theta\rangle,
\end{equation}
where we have used the fact that $\Theta^2\phi=-\phi$ for any vector $\phi$. 
Then, we find that the inner product in the left-hand side 
can be written as 
$$
\langle U_{\bf a}^\ast B\varphi, \Theta\varphi\rangle=\langle \varphi, B(U_{\bf a}\varphi^\Theta)\rangle
=\langle \varphi, U_{\bf a}(B\varphi)^\Theta\rangle,
$$
where we have used the relation (\ref{commutUThetaB}) in Lemma~\ref{lem:commute} to get the second equality. 
Because of the minus sign in the left-hand side of (\ref{ThetapsiThetavarphiinn3}), 
we arrive at the desired result (\ref{orthovarphiUBThetavarphi}).
\end{proof}

These two lemmas imply that the multiplicity of the eigenvalue $\lambda$ of the operator $A$ 
must be even when $\lambda$ satisfies $0<|\lambda|<1$. 
Therefore, if the eigenvalues of $A$ change continuously 
under a continuous variation of the parameters of the Hamiltonian $H$, 
then the parity of the multiplicity of the eigenvalue $\lambda=1$ must be invariant 
under the deformation of the Hamiltonian. 
Thus, we have two possibilities: (i) An even number of eigenvectors of $A$ are lifted from the sector spanned 
by the eigenvectors of $A$ with the eigenvalue $\lambda=1$; (ii) an even number of eigenvectors of $A$ 
with eigenvalue $\lambda \ne 1$ become degenerate with the eigenvectors of $A$ with the eigenvalue $\lambda=1$. 
In both cases, it is enough to prove the continuity of the eigenvalues of the operator $A$ 
under deformation of the Hamiltonian \cite{RSB,Koma2}, in order to establish our main result.  

Let us consider a perturbation $\delta H$ for the odd time-reversal symmetric Hamiltonian $H$. 
We assume that the perturbation $\delta H$ is odd time-reversal symmetric with respect to 
the same time-reversal transformation $\Theta$ as the unperturbed Hamiltonian. 
We also assume that the range of the hopping integrals in the perturbed Hamiltonian $\delta H$ is finite. 
Therefore, the norm $\Vert \delta H\Vert$ of the Hamiltonian $\delta H$ is finite. 
Clearly, the total Hamiltonian $H'=H+\delta H$ 
is odd time-reversal symmetric, too.   
We assume that the Fermi level of the Hamiltonian $H'$ still lies in the spectral gap of $H'$. 
Using a contour integral, the projection onto the Fermi sea for the Hamiltonian $H'$ is written 
\begin{equation}
\label{contourPF}
P_{\rm F}'=\frac{1}{2\pi i}\oint dz \frac{1}{z-H'},
\end{equation}
where the contour encloses all the spectrum of $H'$ below the Fermi level. 

Similarly to the operator $A$ of (\ref{def:A}), we introduce 
$$
A'=P_{\rm F}'-U_{\bf a}P_{\rm F}'U_{\bf a}^\ast.
$$
Then, the difference between $A$ and $A'$ is written as
$$
A'-A=(P_{\rm F}'-P_{\rm F})-U_{\bf a}(P_{\rm F}'-P_{\rm F})U_{\bf a}^\ast. 
$$
In order to evaluate this, it is enough to estimate $P_{\rm F}'-P_{\rm F}$.  
Using the contour integral, we have 
\begin{align*}
P_{\rm F}'-P_{\rm F}&=\frac{1}{2\pi i}\oint dz 
\left[\frac{1}{z-H'}-\frac{1}{z-H}\right]\\
&=\frac{1}{2\pi i}\oint dz 
\frac{1}{z-H'}\delta H\frac{1}{z-H}.
\end{align*}
Because of the assumption of the spectral gap, the norm of the right-hand side can be bounded by 
the norm of the perturbation $\delta H$ with some positive constant. 
Therefore, the difference $A'-A$ is bounded by 
the norm, too. In consequence, the operator $A$ is continuous with respect to 
the norm of the perturbation $\delta H$. 
Using the min-max principle,\footnote{For the min-max principle, see, e.g., Chap.~XIII.1 of the book of 
\cite{RSIV}.} 
we obtain the desired result that the nonzero eigenvalue $\lambda$ 
of the operator $A$ is continuous with respect to the norm of the perturbation $\delta H$. 

To summarize, we obtain: 

\begin{thm}
\label{thm:main}
Suppose that the Hamiltonian $H$ is odd time-reversal symmetric 
with respect to the transformation $\Theta$ of (\ref{timereversal}) whose unitary operator 
$U^\Theta$ satisfies the condition (\ref{assump:UTheta}) for the unitary operator $U_{\bf a}$ of 
the operator $A$ of (\ref{def:A}).  
We assume that the Fermi level lies in a spectral gap of the Hamiltonian $H$. 
Then, the $\ze_2$ index of (\ref{ze2index}) is continuous with respect to 
the norm of any perturbation if the perturbed Hamiltonian preserves the same 
time-reversal symmetry as the unperturbed Hamiltonian. 
In other words, as long as the spectral gap does not close, 
the $\ze_2$ index is robust against any perturbation which preserves 
the same time-reversal symmetry as the unperturbed Hamiltonian.  
\end{thm}
\bigskip

\begin{rem}
(i) When the Fermi level lies in the localization regime, the same statement is still valid. 
But, we need to prove the localization of the particles separately. 
More precisely, we need a decay estimate 
for the resolvent $(z-H)^{-1}$, $z\in\co$. See, e.g., \cite{RSB,Koma2} for the homotopy argument in such a case. 
See also \cite{PLB} which discusses the conditions under which
both the quantization and homotopy invariance of the Chern
number in four or higher dimensions hold in the presence of strong disorder. 
\smallskip

\noindent
(ii) Since the $\ze_2$ index is given by the dimension of the kernel of the operator as in (\ref{ze2index}), 
one might think that the $\ze_2$ index may be written in terms of an integral of some connection, 
similarly to the Chern number. In fact, 
the $\ze_2$ index which was defined by Kane and Mele 
can be written in terms of an integral of the same connection as that of the Chern number 
over one-half of the Brillouin zone 
for translationally invariant systems \cite{FuKane,MooreBalents,EssinMoore}. 
(See also related articles \cite{FuKane2,Roy,FukuiHatsugai,LeeRyu}.) 
The explicit computation of the $\ze_2$ index with the use of the integral formula is given 
in Appendix~\ref{RelationZ2}.  
\smallskip

\noindent
(iii)  As mentioned in Introduction, Schulz-Baldes \cite{SB,GSB} defined 
$\ze_2$-indices for general odd symmetric Fredholm operators 
and treated the $\ze_2$ index for the Kane-Mele model as an example. 
(See also \cite{FukuiFujiwara,DNSB} for related articles.)
But his approach is different from ours. In fact, he defined the $\ze_2$ index 
by the parity of ${\rm dim}\; {\rm ker}\; T$ for the Fredholm operator 
$T:=P_{\rm F}U_{\bf a}P_{\rm F}+(1-P_{\rm F})$. From Eq.~(3.2a) in the proof of Proposition~3.1 in \cite{ASS2}, 
the dimension of the kernel of $T$ coincides with ours as  
${\rm dim}\; {\rm ker}\; T={\rm dim}\; {\rm ker}\;(A-1)$ 
for the operator $A$ of (\ref{def:A}). 
On the one hand, the operator $T$ is noncompact and contains the essential spectrum.
On the other hand, the operator $A$ has only the discrete spectrum 
with finite multiplicity except for zero, because $A$ is compact 
which follows from the fact that $A^3$ is trace class. 
Therefore, a homotopy argument for $A$ is much easier to handle than that for $T$. 
It should be noted that 
the homotopy argument is indispensable for defining the $\ze_2$ index for disordered systems. 
\smallskip

\noindent
(iv)  Hastings and Loring \cite{HL1,HL2,LH} proposed an alternative noncommutative approach to 
the $\ze_2$ index of topological insulators.  
\end{rem}

\medskip

\section{The $\ze_2$ Index of the Kane-Mele Model}
\label{KMmodel}

As a nontrivial example, we demonstrate in this section 
that the $\ze_2$ index for the Kane-Mele model can be computed from the Chern number 
for the corresponding Haldane model with a specific flux pattern and vanishing staggered sublattice potential. 
Both Hamiltonians of the Kane-Mele and the Haldane models have a single spectral gap except at critical points. 
In what follows, we assume that the Fermi level lies in the spectral gap. 
We first note that the condition (\ref{assump:UTheta}) holds for the Kane-Mele model. 
Because of the continuity of the $\ze_2$ index, it is sufficient to 
treat the simplest case without either disorder or the Rashba term in the Hamiltonian. 
Namely, when continuously switching on those terms, 
the $\ze_2$ index does not change as long as the spectral gap does not close. 
As is well known, the Hamiltonian of the Kane-Mele model is decoupled into 
two independent Haldane models \cite{Haldane} when the Rashba term in the Hamiltonian vanishes. 
Although each of the Haldane model shows the nonvanishing Chern number, 
the total Chern number for the decoupled Haldane models is vanishing 
because of the time-reversal symmetry. 
This can be seen as follows. 
First note that the one-to-one mapping $U_{\bf a}\Theta$ sends an eigenvector of the operator $A$ 
with the eigenvalue $\lambda=1$ to that of $A$ with $\lambda=-1$. 
Then, since the difference between the multiplicities of $\lambda=1$ and $\lambda=-1$ gives 
the total Chern number, it follows that the total Chern number is vanishing. 

By contrast, the $\ze_2$ index is given by the parity of the multiplicity of $\lambda=1$.  
Therefore, from these observations, 
we reach a conclusion: 
When the Rashba coupling is 
weak so that the spectral gap does not vanish, 
the $\ze_2$ index for the Kane-Mele model is equal to the parity of the Chern number 
for either one of the decoupled Haldane models. 
Namely, we have 
\begin{equation}
{\rm Ind}_2(P_{\rm F},U_{\bf a} P_{\rm F}U_{\bf a}^\ast)=I_{\rm Ch}\ \ \mbox{modulo}\ 2,
\end{equation}
where we have written $I_{\rm Ch}$ for the Chern number for the single Haldane model.  
In Appendix~\ref{ChernHaldane}, we show that $I_{\rm Ch}=1$ for the single Haldane model in the present setting. 
Therefore, the $\ze_2$ index for the Kane-Mele model is equal to $1$.  

\medskip

\appendix

\section{The Index Theorem and the Chern Number}
\label{IndChern}

Using the method of noncommutative geometry \cite{Connes,BVS,ASS,AG}, 
one can show that the quantity, ${\rm Tr}\; A^3$, in the right-hand side of the index formula (\ref{index}) 
is equal to the Chern number. 
In this appendix, we do not require that the Hamiltonian $H$ is odd time-reversal symmetric. 

By using a function $\theta_{\bf a}({\bf u})$ which is the angle of sight from ${\bf a}$ to ${\bf u}$, 
the unitary operator $U_{\bf a}({\bf u})$ of (\ref{def:Ua}) can be written as
$U_{\bf a}({\bf u})=\exp[i\theta_{\bf a}({\bf u})]$.
We introduce an orthonormal complete system of functions,
$$
\left(\chi_{\bf u}^{\alpha,\mu}\right)_\beta({\bf n},\nu):=
\begin{cases} 1, & \text{${\bf n}={\bf u}$, $\beta=\alpha$ and $\nu=\mu$};\\
0, & \text{otherwise}.
\end{cases}
$$
Namely, the function $\chi_{\bf u}^{\alpha,\mu}$ has the nonvanishing value $1$ only for 
the spin $\alpha$ and the $\mu$-th orbit at the site ${\bf u}$. 
One has the matrix elements, 
$$
\langle\chi_{\bf u}^{\alpha,\mu},(P_{\rm F}-U_{\bf a}P_{\rm F}U_{\bf a}^\ast)\chi_{\bf v}^{\beta,\nu}\rangle 
=\tau_{{\bf u},{\bf v}}\mathcal{T}_{{\bf u},{\bf v}}^{\alpha,\mu;\beta,\nu}, 
$$
where
$$
\tau_{{\bf u},{\bf v}}=1-\exp[i\theta_{\bf a}({\bf u})-i\theta_{\bf a}({\bf v})]
$$
and 
\begin{equation}
\label{def:mathcalT}
\mathcal{T}_{{\bf u},{\bf v}}^{\alpha,\mu;\beta,\nu}
:=\langle\chi_{\bf u}^{\alpha,\mu},P_{\rm F}\chi_{\bf v}^{\beta,\nu}\rangle.
\end{equation}
Using the complete basis, the index of (\ref{index}) is written  
\begin{align*}
{\rm Ind}(P_{\rm F},U_{\bf a}P_{\rm F}U_{\bf a}^\ast)&={\rm Tr}\; (P_{\rm F}-U_{\bf a}P_{\rm F}U_{\bf a}^\ast)^3\\
&=\sum_{{\bf u},{\bf v},{\bf w}\in\ze^2}\tau_{{\bf u},{\bf v}}\tau_{{\bf v},{\bf w}}\tau_{{\bf w},{\bf u}}
\mathcal{S}_{{\bf u},{\bf v},{\bf w},{\bf u}},
\end{align*}
where 
\begin{equation}
\label{def:mathcalS}
\mathcal{S}_{{\bf u},{\bf v},{\bf w},{\bf u}}
=\sum_{\alpha,\beta,\gamma} \sum_{\mu,\nu,\xi}
\mathcal{T}_{{\bf u},{\bf v}}^{\alpha,\mu;\beta,\nu}
\mathcal{T}_{{\bf v},{\bf w}}^{\beta,\nu;\gamma,\xi}
\mathcal{T}_{{\bf w},{\bf u}}^{\gamma,\xi;\alpha,\mu}.
\end{equation} 
As shown in Proposition~3.8 in \cite{ASS}, the index ${\rm Ind}(P_{\rm F},U_{\bf a}P_{\rm F}U_{\bf a}^\ast)$ 
is independent of the positioning ${\bf a}$ of the flux tube. Relying on this fact, the index is written  
$$
{\rm Ind}(P_{\rm F},U_{\bf a}P_{\rm F}U_{\bf a}^\ast)
=\frac{1}{|\Lambda^\ast|}\sum_{{\bf a}\in\Lambda^\ast}
\sum_{{\bf u},{\bf v},{\bf w}\in\ze^2}\tau_{{\bf u},{\bf v}}\tau_{{\bf v},{\bf w}}\tau_{{\bf w},{\bf u}}
\mathcal{S}_{{\bf u},{\bf v},{\bf w},{\bf u}},
$$
where $\Lambda^\ast$ is a finite rectangular box which is a subset of the dual square lattice $(\ze^2)^\ast$. 
Following Elgart, Graf and Schenker \cite{EGS}, one has 

\begin{lem}
There exists a sequence of finite rectangular boxes $\Lambda\subset\ze^2$ such that   
$$
{\rm Ind}(P_{\rm F},U_{\bf a}P_{\rm F}U_{\bf a}^\ast)
=\lim_{\Lambda\uparrow\ze^2}\frac{1}{|\Lambda|}\sum_{{\bf u}\in\Lambda}
\sum_{{\bf v},{\bf w}\in\ze^2}\sum_{{\bf a}\in(\ze^2)^\ast}
\tau_{{\bf u},{\bf v}}\tau_{{\bf v},{\bf w}}\tau_{{\bf w},{\bf u}}
\mathcal{S}_{{\bf u},{\bf v},{\bf w},{\bf u}}.
$$
\end{lem}

Further, in order to compute the right-hand side of the index, 
we apply the Connes' area formula \cite{ConnesArea,ASS,AG}, 
\begin{align*}
\sum_{{\bf a}\in(\ze^2)^\ast}
\tau_{{\bf u},{\bf v}}\tau_{{\bf v},{\bf w}}\tau_{{\bf w},{\bf u}}
&=2\pi i ({\bf v}-{\bf w})\times({\bf w}-{\bf u})\\
&=2\pi i[(v_2-w_2)(w_1-u_1)-(v_1-w_1)(w_2-u_2)]. 
\end{align*}
In consequence, the index is written 
$$
{\rm Ind}(P_{\rm F},U_{\bf a}P_{\rm F}U_{\bf a}^\ast)
=\lim_{\Lambda\uparrow\ze^2}\frac{2\pi i}{|\Lambda|}\sum_{{\bf u}\in\Lambda}
\sum_{{\bf v},{\bf w}\in\ze^2}
({\bf v}-{\bf w})\times({\bf w}-{\bf u})
\mathcal{S}_{{\bf u},{\bf v},{\bf w},{\bf u}}.
$$
Note that 
\begin{align*}
&(v_2-w_2)(w_1-u_1)\langle\chi_{\bf u}^{\alpha,\mu},P_{\rm F}\chi_{\bf v}^{\beta,\nu}\rangle
\langle\chi_{\bf v}^{\beta,\nu},P_{\rm F}\chi_{\bf w}^{\gamma,\xi}\rangle
\langle\chi_{\bf w}^{\gamma,\xi},P_{\rm F}\chi_{\bf u}^{\alpha,\mu}\rangle\\
&=\langle\chi_{\bf u}^{\alpha,\mu},P_{\rm F}\chi_{\bf v}^{\beta,\nu}\rangle
\langle\chi_{\bf v}^{\beta,\nu},[X_2,P_{\rm F}]\chi_{\bf w}^{\gamma,\xi}\rangle
\langle\chi_{\bf w}^{\gamma,\xi},[X_1,P_{\rm F}]\chi_{\bf u}^{\alpha,\mu}\rangle
\end{align*}
and 
\begin{align*}
&(v_1-w_1)(w_2-u_2)\langle\chi_{\bf u}^{\alpha,\mu},P_{\rm F}\chi_{\bf v}^{\beta,\nu}\rangle
\langle\chi_{\bf v}^{\beta,\nu},P_{\rm F}\chi_{\bf w}^{\gamma,\xi}\rangle
\langle\chi_{\bf w}^{\gamma,\xi},P_{\rm F}\chi_{\bf u}^{\alpha,\mu}\rangle\\
&=\langle\chi_{\bf u}^{\alpha,\mu},P_{\rm F}\chi_{\bf v}^{\beta,\nu}\rangle
\langle\chi_{\bf v}^{\beta,\nu},[X_1,P_{\rm F}]\chi_{\bf w}^{\gamma,\xi}\rangle
\langle\chi_{\bf w}^{\gamma,\xi},[X_2,P_{\rm F}]\chi_{\bf u}^{\alpha,\mu}\rangle, 
\end{align*}
where the operators $X_j$ are given by 
$(X_j\varphi)({\bf n})=n_j\varphi({\bf n})$, $j=1,2$, 
for $\varphi({\bf n})\in\co^r\otimes\co^2$. Combining these, (\ref{def:mathcalT}) and (\ref{def:mathcalS}), 
we obtain 
$$
{\rm Ind}(P_{\rm F},U_{\bf a}P_{\rm F}U_{\bf a}^\ast)
=-2\pi i\lim_{\Lambda\uparrow\ze^2}\frac{1}{|\Lambda|}{\rm Tr}\;\chi_\Lambda
P_{\rm F}[[X_1,P_{\rm F}],[X_2,P_{\rm F}]]\chi_\Lambda,
$$
where $\chi_\Lambda$ is the characteristic function of the finite lattice $\Lambda$, i.e., 
$$
\chi_{\Lambda}({\bf n})=
\begin{cases} 1, & \text{${\bf n}\in{\Lambda}$};\\
0, & \text{otherwise}.
\end{cases}
$$
Further, the index can be expressed as \cite{ASS,EGS,Koma2} 
$$
{\rm Ind}(P_{\rm F},U_{\bf a}P_{\rm F}U_{\bf a}^\ast)
=-2\pi i{\rm Tr}\;P_{\rm F}[[\vartheta_{1,{\bf a}},P_{\rm F}],[\vartheta_{2,{\bf a}},P_{\rm F}]]
$$
in terms of the step functions,  
$$
\vartheta_{j,{\bf a}}({\bf n})=
\begin{cases} 1, & \text{$n_j-a_j\ge 0$};\\
0, & \text{$n_j-a_j<0$,}
\end{cases}
$$
for $j=1,2$.
Using a contour integral as in (\ref{contourPF}), we have 
\begin{align*}
[\vartheta_{j,{\bf a}},P_{\rm F}]&=\frac{1}{2\pi i}\oint dz
\left(\vartheta_{j,{\bf a}}\frac{1}{z-H}-\frac{1}{z-H}\vartheta_{j,{\bf a}}\right)\\
&=-\frac{1}{2\pi i}\oint dz
\frac{1}{z-H}[H,\vartheta_{j,{\bf a}}]\frac{1}{z-H}\\
&=\frac{i}{2\pi i}\oint dz \frac{1}{z-H}J_{j,{\bf a}}\frac{1}{z-H},
\end{align*}
where we have introduced \cite{Koma5} the local current operators 
$J_{j,{\bf a}}=i[H,\vartheta_{j,{\bf a}}]$ 
for $j=1,2$.
Substituting this into the expression of the index, one has 
\begin{multline}
{\rm Ind}(P_{\rm F},U_{\bf a}P_{\rm F}U_{\bf a}^\ast)\\
=\frac{1}{2\pi i}\oint dz\oint dz'\;{\rm Tr}\; P_{\rm F}
\frac{1}{z-H}J_{1,{\bf a}}\frac{1}{z-H}\frac{1}{z'-H}J_{2,{\bf a}}\frac{1}{z'-H}
-(1\leftrightarrow 2).
\end{multline}
Consider the system on the finite lattices $\Lambda$. 
When there exists a uniform spectral gap above the Fermi level with respect to the size of 
the finite lattices $\Lambda$, the index is written as
\begin{multline}
\label{ChernNumberFinite}
{\rm Ind}(P_{\rm F},U_{\bf a}P_{\rm F}U_{\bf a}^\ast)\\
={2\pi i}\lim_{\Lambda\uparrow \ze^2}\sum_{i:E_i< E_{\rm F}}\sum_{j:E_j>E_{\rm F}}
\frac{1}{(E_i-E_j)^2}
\langle \Phi_i,J_{1,{\bf a}}\Phi_j\rangle \langle \Phi_j,J_{2,{\bf a}}\Phi_i\rangle
-(1\leftrightarrow 2)
\end{multline}
in terms of the eigenvectors $\Phi_i$ of the finite-volume Hamiltonian $H$ 
with the energy eigenvalue $E_i$. Here, $E_{\rm F}$ is the Fermi energy. 
The right-hand side expresses the Hall conductance which is quantized to an integer because 
the analytical index 
${\rm Ind}(P_{\rm F},U_{\bf a}P_{\rm F}U_{\bf a}^\ast)$ of the operator 
in the left-hand side takes an integer value by definition. 

\medskip

\section{The Chern Number of the Haldane Model}
\label{ChernHaldane}

In this Appendix, we present a detailed exposition of the calculation of the Chern number for the Haldane model. 
The original Haldane model is defined on a honeycomb lattice. 
For technical reasons, however, we need to define it on the square lattice $\ze^2$ 
in the present setting.\footnote{
Note that an arbitrary two-dimensional tight-binding model can be mapped onto the model 
on $\ze^2$ with suitably chosen hopping integrals. 
} 
For this purpose, we introduce two orbitals labeled by $a$ and $b$ at each lattice site ${\bf n}\in\ze^2$. 
We write 
$$
\varphi({\bf n})=\left(\begin{matrix} \varphi^a({\bf n}) \\ 
\varphi^b({\bf n})
\end{matrix}\right)\in \co^2
$$
for the wavefunction $\varphi\in \ell^2(\ze^2,\co^2)$ at the site ${\bf n}=(n_1,n_2)$. 
In order to prove the statement in Sec. \ref{KMmodel}, it suffices to consider the Haldane model 
with a pure-imaginary hopping term between second-neighbor sites and the vanishing staggered sublattice potential. 
In this case, the Schr\"odinger equation reads 
\begin{multline*}
t[\varphi^b(n_1-1,n_2)+\varphi^b(n_1,n_2)+\varphi^b(n_1,n_2+1)]\\
+it'[\varphi^a(n_1-1,n_2-1)-\varphi^a(n_1,n_2-1)+\varphi^a(n_1+1,n_2)-\varphi^a(n_1+1,n_2+1)\\
+\varphi^a(n_1,n_2+1)-\varphi^a(n_1-1,n_2)]=E\varphi^a(n_1,n_2)
\end{multline*}
and 
\begin{multline*}
t[\varphi^a(n_1,n_2)+\varphi^a(n_1,n_2-1)+\varphi^a(n_1+1,n_2)]\\
+it'[\varphi^b(n_1-1,n_2)-\varphi^b(n_1-1,n_2-1)+\varphi^b(n_1,n_2-1)-\varphi^b(n_1+1,n_2)\\
+\varphi^b(n_1+1,n_2+1)-\varphi^b(n_1,n_2+1)]=E\varphi^b(n_1,n_2),
\end{multline*}
where the hopping integrals, $t$ and $t'$, are real constants, and $E$ is the energy eigenvalue.  
In the following, we will treat only the case with $t'>0$. 

In order to use the formula (\ref{ChernNumberFinite}) for computing the Chern number, 
we consider a finite $L_1\times L_2$ rectangular box with the periodic boundary conditions in $\ze^2$. 
We set 
$$
\varphi({\bf n})=\left(\begin{matrix} \varphi^a({\bf n}) \\ 
\varphi^b({\bf n})
\end{matrix}\right)=\frac{1}{\sqrt{L_1L_2}}\exp[i{\bf k}\cdot{\bf n}]
\left(\begin{matrix} \tilde{\varphi}^a({\bf k}) \\ 
\tilde{\varphi}^b({\bf k})
\end{matrix}\right),
$$
where ${\bf k}=(k_1,k_2)$ is the wave number vector, and $\tilde{\varphi}^a$ and $\tilde{\varphi}^b$ 
are functions of ${\bf k}$. 
Substituting this into the above Schr\"odinger equation, one has 
$$
H({\bf k})\left(\begin{matrix} \tilde{\varphi}^a({\bf k}) \\ 
\tilde{\varphi}^b({\bf k})
\end{matrix}\right)=E\left(\begin{matrix} \tilde{\varphi}^a({\bf k}) \\ 
\tilde{\varphi}^b({\bf k})
\end{matrix}\right),
$$
where 
$$
H({\bf k})=
\left(\begin{matrix}
-\Delta({\bf k}) & \overline{\Gamma({\bf k})} \\
\Gamma({\bf k}) & \Delta({\bf k})
\end{matrix}\right)
$$
with 
$
\Delta({\bf k})=2t'[\sin k_1+\sin k_2-\sin(k_1+k_2)]
$
{and}
\begin{equation}
\label{def:Gamma} 
\Gamma({\bf k})=t(1+e^{ik_1}+e^{-ik_2}).
\end{equation}
The energy eigenvalues are then given by 
$
E=E_\pm({\bf k})=\pm \mathcal{E}({\bf k}) 
$
with 
$
\mathcal{E}({\bf k})=\sqrt{\Delta({\bf k})^2+|\Gamma({\bf k})|^2}.
$

Let us consider the case of $\Gamma({\bf k})=0$. 
In the Brillouin zone, $(-\pi,\pi]\times(-\pi,\pi]$, 
there are only two wave-number vectors that satisfy this condition. 
They are given by 
$
(k_1,k_2)=\pm (2\pi/3,2\pi/3). 
$
For these two points, one has 
$
\Delta({\bf k})=\pm 3\sqrt{3}t' 
$
for the positive and negative wave numbers, respectively, because we have assumed $t'>0$. 
Thus, there exists a nonvanishing spectral gap between the upper and lower energy bands 
for $t\ne 0$ and $t'\ne 0$.  

The eigenvector of the lower band is given by \cite{MurakamiNagaosa}
$$
f_-({\bf k})=\frac{1}{\sqrt{2\mathcal{E}({\bf k})[\mathcal{E}({\bf k})+\Delta({\bf k})]}}
\left(\begin{matrix}
\mathcal{E}({\bf k})+\Delta({\bf k}) \\ -\Gamma({\bf k})
\end{matrix}\right).
$$
But, the normalization factor of this vector becomes infinite at $(k_1,k_2)=-(2\pi/3,2\pi/3)$ 
because $\Delta({\bf k})=-\mathcal{E}({\bf k})$ at this point. 
An alternative expression of the eigenvector is given by \cite{MurakamiNagaosa}
$$
g_-({\bf k})=\frac{1}{\sqrt{2\mathcal{E}({\bf k})[\mathcal{E}({\bf k})-\Delta({\bf k})]}}
\left(\begin{matrix}
-\overline{\Gamma({\bf k})} \\ \mathcal{E}({\bf k})-\Delta({\bf k}) \end{matrix}\right).
$$
In this case, the normalization factor of this vector becomes infinite at $(k_1,k_2)=(2\pi/3,2\pi/3)$. 
The relation between these two vectors is given by the gauge transformation, 
$f_-({\bf k})=e^{i\eta({\bf k})}g_-({\bf k})$, for $(k_1,k_2)\ne\pm (2\pi/3,2\pi/3)$, 
where the angle function $\eta({\bf k})$ is given by  
\begin{equation}
\label{def:eta}
e^{i\eta({\bf k})}=-\frac{\Gamma({\bf k})}{|\Gamma({\bf k})|}.
\end{equation}
We choose the eigenvector as 
$$
\tilde{\varphi}_-^{\rm up}({\bf k})=\left(\begin{matrix} \tilde{\varphi}_-^a({\bf k}) \\ 
\tilde{\varphi}_-^b({\bf k})
\end{matrix}\right)
=f_-({\bf k})\quad \mbox{for \ } (k_1,k_2)\in \mathcal{U}^{\rm up}:=(-\pi,\pi]\times[0,\pi]
$$
and 
$$
\tilde{\varphi}_-^{\rm low}({\bf k})=\left(\begin{matrix} \tilde{\varphi}_-^a({\bf k}) \\ 
\tilde{\varphi}_-^b({\bf k})
\end{matrix}\right)
=g_-({\bf k})\quad \mbox{for \ } (k_1,k_2)\in \mathcal{U}^{\rm low}:=(-\pi,\pi]\times[-\pi,0]. 
$$
Using these eigenvectors and the translational invariance, the Chern number $I_{\rm Ch}$ 
which is given by the right-hand side of (\ref{ChernNumberFinite}) is written as
\begin{equation}
\label{IndIntuplow}
I_{\rm Ch}=\frac{i}{2\pi}\left(K^{\rm up}+K^{\rm low}\right)
\end{equation}
with 
\begin{multline}
\label{Iup}
K^{\rm up}=\int_{\mathcal{U}^{\rm up}}dk_1dk_2
\frac{\langle\tilde{\varphi}_-^{\rm up}({\bf k}),J_1({\bf k})\tilde{\varphi}_+^{\rm up}({\bf k})\rangle
\langle\tilde{\varphi}_+^{\rm up}({\bf k}),J_2({\bf k})\tilde{\varphi}_-^{\rm up}({\bf k})\rangle}
{[E_+({\bf k})-E_-({\bf k})]^2}\\-(1\leftrightarrow 2)
\end{multline}
and 
\begin{multline}
K^{\rm low}=\int_{\mathcal{U}^{\rm low}}dk_1dk_2
\frac{\langle\tilde{\varphi}_-^{\rm low}({\bf k}),J_1({\bf k})\tilde{\varphi}_+^{\rm low}({\bf k})\rangle
\langle\tilde{\varphi}_+^{\rm low}({\bf k}),J_2({\bf k})\tilde{\varphi}_-^{\rm low}({\bf k})\rangle}
{[E_+({\bf k})-E_-({\bf k})]^2}\\-(1\leftrightarrow 2),
\end{multline}
where $\tilde{\varphi}_+^{\rm up}({\bf k})$ and $\tilde{\varphi}_+^{\rm low}({\bf k})$ 
are the corresponding eigenvectors for the upper band, and the current operators are given by 
$$
J_j({\bf k})=\frac{\partial}{\partial k_j}H({\bf k}),\quad j=1,2. 
$$
By differentiating 
$$
\langle\tilde{\varphi}_-^{\rm up}({\bf k}),H({\bf k})\tilde{\varphi}_+^{\rm up}({\bf k})\rangle=0,  
$$
one has 
\begin{multline*}
\langle\frac{\partial}{\partial k_j}
\tilde{\varphi}_-^{\rm up}({\bf k}),H({\bf k})\tilde{\varphi}_+^{\rm up}({\bf k})\rangle
+\langle\tilde{\varphi}_-^{\rm up}({\bf k}),\frac{\partial}{\partial k_j}H({\bf k})
\tilde{\varphi}_+^{\rm up}({\bf k})\rangle\\
+\langle\tilde{\varphi}_-^{\rm up}({\bf k}),H({\bf k})
\frac{\partial}{\partial k_j}\tilde{\varphi}_+^{\rm up}({\bf k})\rangle=0.  
\end{multline*}
{From} this, the definition of the current $J_j({\bf k})$ and the identity, 
\begin{equation}
\label{diffnormID}
\langle\frac{\partial}{\partial k_j}
\tilde{\varphi}_-^{\rm up}({\bf k}),\tilde{\varphi}_+^{\rm up}({\bf k})\rangle
+\langle\tilde{\varphi}_-^{\rm up}({\bf k}),
\frac{\partial}{\partial k_j}\tilde{\varphi}_+^{\rm up}({\bf k})\rangle=0,
\end{equation}
one obtains 
$$
\langle\tilde{\varphi}_-^{\rm up}({\bf k}),J_j({\bf k})
\tilde{\varphi}_+^{\rm up}({\bf k})\rangle=
[E_-({\bf k})-E_+({\bf k})]\langle\frac{\partial}{\partial k_j}
\tilde{\varphi}_-^{\rm up}({\bf k}),\tilde{\varphi}_+^{\rm up}({\bf k})\rangle.
$$
In the same way, 
$$
\langle\tilde{\varphi}_+^{\rm up}({\bf k}),J_j({\bf k})
\tilde{\varphi}_-^{\rm up}({\bf k})\rangle=
[E_-({\bf k})-E_+({\bf k})]\langle
\tilde{\varphi}_+^{\rm up}({\bf k}),\frac{\partial}{\partial k_j}
\tilde{\varphi}_-^{\rm up}({\bf k})\rangle.
$$
Substituting these into the integral $K^{\rm up}$ of (\ref{Iup}) 
and using the identities (\ref{diffnormID}) and
$$
\langle\frac{\partial}{\partial k_j}
\tilde{\varphi}_-^{\rm up}({\bf k}),\tilde{\varphi}_-^{\rm up}({\bf k})\rangle
+\langle\tilde{\varphi}_-^{\rm up}({\bf k}),
\frac{\partial}{\partial k_j}\tilde{\varphi}_-^{\rm up}({\bf k})\rangle=0,
$$
one has 
\begin{align}
\label{IupInt}
K^{\rm up}&=\int_{\mathcal{U}^{\rm up}}dk_1dk_2 
\langle \frac{\partial}{\partial k_1}\tilde{\varphi}_-^{\rm up}({\bf k}),
\tilde{\varphi}_+^{\rm up}({\bf k})\rangle 
\langle \tilde{\varphi}_+^{\rm up}({\bf k}),\frac{\partial}{\partial k_2}
\tilde{\varphi}_-^{\rm up}({\bf k})\rangle-(1\leftrightarrow 2)\\ \nonumber
&=\int_{\mathcal{U}^{\rm up}}dk_1dk_2 
\left[\langle \frac{\partial}{\partial k_1}\tilde{\varphi}_-^{\rm up}({\bf k}),
\frac{\partial}{\partial k_2}
\tilde{\varphi}_-^{\rm up}({\bf k})\rangle
-\langle \frac{\partial}{\partial k_2}\tilde{\varphi}_-^{\rm up}({\bf k}),
\frac{\partial}{\partial k_1}
\tilde{\varphi}_-^{\rm up}({\bf k})\rangle\right]\\ \nonumber
&=\int_{\mathcal{U}^{\rm up}}dk_1dk_2 
\left[\frac{\partial}{\partial k_1}\langle \tilde{\varphi}_-^{\rm up}({\bf k}),
\frac{\partial}{\partial k_2}
\tilde{\varphi}_-^{\rm up}({\bf k})\rangle
-\frac{\partial}{\partial k_2}\langle \tilde{\varphi}_-^{\rm up}({\bf k}),
\frac{\partial}{\partial k_1}
\tilde{\varphi}_-^{\rm up}({\bf k})\rangle\right]\\ \nonumber 
&=\int_{-\pi}^\pi dk_1\left[
\langle \tilde{\varphi}_-^{\rm up}(k_1,0),
\frac{\partial}{\partial k_1}
\tilde{\varphi}_-^{\rm up}(k_1,0)\rangle
-\langle \tilde{\varphi}_-^{\rm up}(k_1,\pi),
\frac{\partial}{\partial k_1}
\tilde{\varphi}_-^{\rm up}(k_1,\pi)\rangle\right].
\end{align}
In the same way, 
\begin{align}
\label{IlowLint}
K^{\rm low}&=-\int_{-\pi}^\pi dk_1\langle \tilde{\varphi}_-^{\rm low}(k_1,0),
\frac{\partial}{\partial k_1}
\tilde{\varphi}_-^{\rm low}(k_1,0)\rangle
\\ \nonumber
&+\int_{-\pi}^\pi dk_1\langle \tilde{\varphi}_-^{\rm low}(k_1,-\pi),
\frac{\partial}{\partial k_1}
\tilde{\varphi}_-^{\rm low}(k_1,-\pi)\rangle.
\end{align}

On the other hand, one has 
$$
\frac{\partial}{\partial k_1}\tilde{\varphi}_-^{\rm up}({\bf k})
=ie^{i\eta({\bf k})}\left[\frac{\partial\eta({\bf k})}{\partial k_1}\right]\tilde{\varphi}_-^{\rm low}({\bf k})
+e^{i\eta({\bf k})}\frac{\partial}{\partial k_1}\tilde{\varphi}_-^{\rm low}({\bf k}).
$$
from the gauge transformation, 
$\tilde{\varphi}_-^{\rm up}({\bf k})=e^{i\eta({\bf k})}\tilde{\varphi}_-^{\rm low}({\bf k})$.
Combining this with (\ref{IupInt}) and (\ref{IlowLint}), one obtains 
\begin{align*}
K^{\rm up}+K^{\rm low}&=\int_{-\pi}^\pi dk_1\frac{\partial}{\partial k_1}\eta(k_1,0)
-\int_{-\pi}^\pi dk_1\frac{\partial}{\partial k_1}\eta(k_1,\pi)=-2\pi i,
\end{align*}
where we have used (\ref{def:Gamma}) and (\ref{def:eta}). 
Substituting this into equation (\ref{IndIntuplow}), the value of the Chern number is calculated as  
$I_{\rm Ch}=1$ for the present Haldane model. 

\section{Relation to the $\ze_2$ index by Kane and Mele}
\label{RelationZ2}

In order to show that the $\ze_2$ index by Kane and Mele is equal to 
the present index (\ref{ze2index}), we recall the integral formula \cite{FuKane} for the $\ze_2$ index. 
We use the formula given by Eq.~(9) in \cite{EssinMoore}, i.e.,   
\begin{equation}
\label{intformula}
D:=\frac{1}{2\pi}\left[\oint_{\partial({\rm EBZ})} d\bm{k}\cdot\bm{\mathcal{A}} 
-\int_{\rm EBZ} dk_1dk_2\; \mathcal{F}\right],
\end{equation}
where $\bm{\mathcal{A}}$ and $\mathcal{F}$ are, respectively, the Berry connection and 
the field strength, and the ``effective Brillouin zone", EBZ, stands for one-half of the Brillouin zone 
with a time-reversal invariant frame.   
Because of the homotopy argument, it suffices to consider the case where 
the Kane-Mele model is decoupled into two independent Haldane models.  

Consider first one of the Haldane models whose lower band carries the Chern number $+1$. 
Following the argument in Appendix~\ref{ChernHaldane}, 
we set ${\rm EBZ}=\mathcal{U}^{\rm up}$. 
Further, we choose 
$$
\bm{\mathcal{A}}_g({\bf k})=-i\langle g_-({\bf k}),\nabla_{\bf k}g_-({\bf k})\rangle
$$
as the connection $\bm{\mathcal{A}}$ in the line integral in (\ref{intformula}). 
For the connection of the field strength $\mathcal{F}$, we choose 
$$
\bm{\mathcal{A}}_f({\bf k})=-i\langle f_-({\bf k}),\nabla_{\bf k}f_-({\bf k})\rangle. 
$$
Here, we stress that the wavefunction $f_-({\bf k})$ has no singularity in $\mathcal{U}^{\rm up}$, 
while the wavefunction $g_-({\bf k})$ has the single singular point in $\mathcal{U}^{\rm up}$.    
Then, the contribution to the index $D$ is written as 
\begin{align*}
D_+&:=\frac{1}{2\pi}\left[\oint_{\partial\mathcal{U}^{\rm up}} d\bm{k}\cdot\bm{\mathcal{A}}_g
-\oint_{\partial\mathcal{U}^{\rm up}} d\bm{k}\cdot\bm{\mathcal{A}}_f\right]\\
&=-\frac{1}{2\pi}\oint_{\partial\mathcal{U}^{\rm up}} d\bm{k}\cdot\nabla_{\bf k}\eta({\bf k})=1
\end{align*}
in the same way as in Appendix~\ref{ChernHaldane}. 

Next consider the other Haldane model which is the time reverse of the above one 
and therefore whose lower band carries the Chern number $-1$.  
The connection $\bm{\mathcal{A}}({\bf k})$ in the line integral can be made nonsingular 
and chosen to be the same as that 
for the field strength.\footnote{Restricting the region to one-half of the Brillouin zone is 
essential to this argument, because one cannot find a nonsingular connection on the whole Brillouin zone.}
It then follows that the corresponding contribution $D_-$ to the index $D$ is vanishing.  

In consequence, the index $D$ of (\ref{intformula}) is equal to $+1$. 
Immediately, one notices that this index $D$ is nothing but the index of the singular point of 
the wavefunction $g_-({\bf k})$. 
The importance of such singularities of wavefunctions on the Brillouin zone has already been 
pointed out by Kohmoto in his early work \cite{Kohmoto} . 
Furthermore, one can easily see that this index $D$ is equal to the present $\ze_2$ index of (\ref{ze2index}) 
by recalling the argument of Appendix~\ref{ChernHaldane}.

\bigskip\bigskip

\noindent
\thanks{\textbf{Acknowledgement:} 
HK was supported in part by JSPS Grants-in-Aid for Scientific Research No. 23740298 and 25400407. 
\bigskip\bigskip



\end{document}